\documentclass{article}

\usepackage[T1]{fontenc}
\usepackage[utf8]{inputenc}
\usepackage{amsmath}
\usepackage{mathtools}
\usepackage{amsthm}
\usepackage{amssymb}
\usepackage{braket}
\usepackage{comment}
\usepackage{hyperref}
\usepackage{authblk}
\usepackage{cite}
\usepackage{bm}
\usepackage{enumerate}
\usepackage[dvipdfmx]{graphicx}
\usepackage{pgfplots}
\pgfplotsset{compat=1.18}

\usepackage{lmodern}
\usepackage{xcolor}

\usepackage{listings}
\theoremstyle{plain}

\newtheorem{lem}{Lemma}

\newtheorem{pro}{Proposition}

\theoremstyle{definition}

\newtheorem{exmp}{Example}

\theoremstyle{remark}

\DeclareMathOperator{\tr}{tr}
\DeclareMathOperator{\Real}{Re}

\DeclarePairedDelimiter{\abs}{\lvert}{\rvert}
\newcommand{\hil}{\mathcal{H}}

\newcommand{\linear}{\mathcal{L}}
\newcommand{\state}{\mathcal{S}}
\newcommand{\id}{\mathcal{I}}

\newcommand{\ketbra}[2]{\ket{#1}\bra{#2}}

\newcommand{\sysa}{\text{anc}}
\newcommand{\sysi}{\text{in}}
\newcommand{\syso}{\text{out}}

\newcommand{\E}{\mathsf{E}}
\newcommand{\T}{\mathsf{T}}

\newcommand\norm[1]{\left\lVert#1\right\rVert}
\newcommand\normOp[1]{\left\lVert#1\right\rVert_{\text{op}}}

\newcommand\iProdHS[2]{\braket{#1|#2}_{\text{HS}}}
\newcommand{\ctil}{\tilde{c}}
\newcommand{\vtil}{\tilde{\mathcal{V}}}

\begin{document}

\title{Fine-Grained Uncertainty Relations for Quantum Testers}

\author{Taihei Kimoto\footnote{kimoto.taihei.22e@st.kyoto-u.ac.jp}}
\affil{Department of Nuclear Engineering, Kyoto University, Kyoto daigaku-katsura\\
Nishikyo-ku, Kyoto 615-8540,
Japan}
\date{}

\maketitle
\begin{abstract}
The uncertainty principle is one of the features of quantum theory.
Fine-grained uncertainty relations (FGURs) are a contemporary interpretation of this principle.
Each FGUR is derived from a scenario where multiple measurements of a quantum state are stochastically performed.
While state measurements are fundamental, measuring quantum processes, namely, completely positive and trace preserving maps, is also crucial both theoretically and practically.
These measurements are mathematically characterized by quantum testers.
In this study, we develop FGURs in terms of quantum testers.
Because state preparation is a type of quantum process, our framework encompasses the conventional case as a special instance.
The generalized FGURs' bound are typically challenging to compute.
Thus, we also provides estimates for these bounds.
Specifically, we explore quantum testers involving maximally entangled states in detail.
Consequently, some FGURs for quantum testers are derived as explicit forms for specific settings.
\end{abstract}

\section{Introduction}
The uncertainty principle is a fundamental aspect of quantum theory.
A well-known inequality that expresses this principle was derived by Robertson\cite{PhysRev.34.163}.
It provides a lower bound for the product of the standard deviations associated with measurements of a quantum state.
Specifically, when this inequality is applied to measurements of a particle's position and momentum, it indicates that the standard deviations cannot be zero simultaneously because the bound becomes a constant.

Thus, Robertson's inequality imposes a nontrivial restriction on the probability distributions of position and momentum.
However, as noted in Ref. \citen{PhysRevLett.50.631}, it has two drawbacks in finite-dimensional systems.
First, the bound of Robertson's inequality can become trivial in certain cases because it depends on the prepared state.
For instance, if the state is an eigenstate of one of the Hermitian operators to be measured, the bound becomes zero, and thus, the inequality does not impose any meaningful restrictions.
Second, the dependence of the standard deviation on measurement outcomes makes it an unsuitable measure for quantifying uncertainty.
Indeed, the standard deviation can change if the corresponding Hermitian operator is scaled by a real number, but this scaling only redefines the outcomes and lacks physical significance.
In Ref. \citen{PhysRevLett.50.631}, Deutsch addressed these issues by using Shannon entropy as a measure of uncertainty.
Maassen and Uffink further refined Deutsch's inequality in Ref. \citen{PhysRevLett.60.1103}, and their result has become foundational in relevant fields.
Uncertainty relations based on informational entropy are known as entropic uncertainty relations (EURs).
EURs are of active research interest not only for their intrinsic value but also for their practical relevance owing to their connection with information processing\cite{RevModPhys.89.015002}.

Other types of uncertainty relations address the issues inherent in Robertson's inequality.
One such example is the fine-grained uncertainty relations (FGURs) introduced in Ref. \citen{doi:10.1126/science.1192065}.
FGURs are sets of inequalities, and each inequality imposes a restriction on a combination of measurement outcomes.
This feature makes FGURs more advantageous than EURs because EURs do not offer such detailed constraints.
Additionally, FGURs can be viewed as more fundamental than EURs because several EURs are derived from FGURs\cite{doi:10.1126/science.1192065,RevModPhys.89.015002}.

Majorization uncertainty relations (MURs) are other contemporary formulations of the uncertainty principle\cite{PhysRevA.84.052117,10.1088/1751-8113/46/27/272002,PhysRevLett.111.230401,PhysRevA.89.052115}.
Majorization is a preorder relation defined for pairs of real vectors.
A key advantage of MURs is that inequalities can be directly obtained by applying Schur-concave functions to them.
A real-valued function defined on a set of real vectors is called Schur-concave if it reverses the order determined by majorization.
For instance, R\'{e}nyi entropy, which generalizes Shannon entropy with a single parameter, is a typical Schur-concave function.
Consequently, MURs can be used to derive EURs formulated with R\'{e}nyi entropy.

Thus, FGURs and MURs can be seen as more refined concepts compared to EURs, and each of them has unique strengths.
FGURs describe restrictions inherent in combinations of outcomes, details not provided by MURs.
Conversely, MURs generally yield stronger EURs than those offered by FGURs\cite{PhysRevA.84.052117,10.1088/1751-8113/46/27/272002,PhysRevLett.111.230401,PhysRevA.89.052115}.
Therefore, both FGURs and MURs are valuable areas of study.

Conventional uncertainty relations, such as EURs, FGURs, and MURs, are primarily formulated for state measurements.
While states are a major focus, other crucial concepts in quantum theory include time evolution induced by unitary operators and measurement processes.
These elements, along with states, can be viewed as specific instances of quantum channels.
Thus, it is natural to extend conventional uncertainty relations to apply to measurements of quantum channels.
Reference \citen{PhysRevResearch.3.023077} presents such generalized uncertainty relations, extending both the EUR from Ref. \citen{Rastegin_2010} (a generalization of the Maassen--Uffink inequality) and the MURs discussed in Refs. \citen{PhysRevA.84.052117,10.1088/1751-8113/46/27/272002,PhysRevLett.111.230401,PhysRevA.89.052115}.
In this study, we perform a similar analysis of FGURs, building on the approach in Ref. \citen{PhysRevResearch.3.023077}.

This paper is organized as follows.
Sections \ref{sec:fg} and \ref{sec:tester} introduce the formulation of FGURs and the framework for measurements of quantum channels, respectively.
Section \ref{sec:result} extends the FGURs, providing a generalization along with a method for estimating the uncertainty bounds, given the difficulty of computing them accurately.
This section also includes concrete examples of the generalized FGURs.
Conclusions and discussion are presented in Sect. \ref{sec:conclusion}.

\section{FGURs} \label{sec:fg}
We identify a quantum system with a Hilbert space $\hil$.
In this study, we assume that Hilbert spaces are finite-dimensional.
The set of all linear operators on $\hil$ is denoted as $\linear(\hil)$.
Quantum states are expressed as positive operators with unit trace, and we denote the set of quantum states on $\hil$ as $\state(\hil)$.
Let $\rho$ be a state in $\state(\hil)$.
A measurement of a state is described by a positive-operator-valued measure (POVM) which is defined as a family of positive operators $\E=\{\E(x)\}_{x\in\Omega}$ such that $\sum_{x\in\Omega}\E(x)=I$, where $\Omega$ is a set of possible outcomes.
We assume that all outcome sets are finite, and hence $\Omega$ is also finite.
If $\rho$ is measured using $\E$, the probability of obtaining an outcome $x$ is given by
\begin{align}
p(x|\rho)
=\tr\left[\E(x)\rho\right],
\end{align}
which is well known as the Born rule.

Let us consider POVMs $\E^{(l)}=\{\E^{(l)}(x^{(l)})\}_{x^{(l)}\in\Omega_{l}},l=1,\ldots,L$, where $L$ is a positive integer and represents the number of measurements.
We assume that $\Omega_{l}\cap\Omega_{l'}=\emptyset$ holds if $l\neq l'$ so that each measurement can be distinguished.
Let $\{r_{l}\}_{l\in[L]_{+}}$ denote a probability distribution, where $[\cdot]_{+}$ is the subset of all integers $\mathbb{Z}$ defined as $[M]_{+}:=\{m\in\mathbb{Z}:1\leq m\leq M\}$ for an arbitrary integer $M$.
We consider the scenario where a state $\rho$ is measured using $\E^{(l)}$ with the probability $r_{l}$.
The probability of obtaining an outcome $x^{(l)}$ is then given by
\begin{align}
r_{l}p_{l}(x^{(l)}|\rho)
=r_{l}\tr\left[\E^{(l)}(x^{(l)})\rho\right].
\end{align}
The FGUR proposed in Ref. \citen{doi:10.1126/science.1192065} is the set of inequalities expressed as follows:
\begin{align}
\mathcal{U}
=\left\{\sum^{L}_{l=1}r_{l}p_{l}(x^{(l)}|\rho)
\leq b(\bm{x}):\bm{x}:=(x^{(1)},\ldots,x^{(L)})\in\prod^{L}_{l=1}\Omega_{l}\right\}. \label{eq:fg}
\end{align}
The bound is defined as follows:
\begin{align}
b(\bm{x})=\max_{\rho}\left\{\sum^{L}_{l=1}r_{l}p_{l}(x^{(l)}|\rho)\right\}, \label{eq:conv_fgur_bound}
\end{align}
where the maximization is performed over all states in $\state(\hil)$.
Each inequality in Eq. \ref{eq:fg} has a clear operational interpretation: it quantifies the limit of the probability of obtaining one of the outcomes in a sequence $\bm{x}\in\prod^{L}_{l=1}\Omega_{l}$.
In other words, each inequality in $\mathcal{U}$ imposes a restriction on the probability associated with a combination of outcomes.
This feature distinguishes FGURs from other uncertainty relations, such as Robertson's inequality, EURs, and MURs, which do not provide this level of detail.
Specifically, consider a combination of outcomes $\bm{x}\in\prod^{L}_{l=1}\Omega_{l}$.
As noted in Ref. \citen{Rastegin2015}, if the strict inequality
\begin{align}
b(\bm{x})
<\sum^{L}_{l=1}r_{l}\max_{\rho}\left\{p_{l}(x^{(l)}|\rho)\right\} \label{eq:inc_cond}
\end{align}
holds, where the meaning of the maximization is the same as that in Eq. (\ref{eq:conv_fgur_bound}), the probabilities
\begin{align}
p_{1}(x^{(1)}|\rho),\ldots,p_{L}(x^{(L)}|\rho)
\end{align}
cannot simultaneously achieve their maximum values, indicating an unavoidable trade-off between the probabilities.
Specifically, if each $\E^{(l)}$ is a projection-valued measure (PVM), the right-hand side of (\ref{eq:inc_cond}) becomes $1$.
Thus, an unavoidable trade-off exists if $b(\bm{x})<1$.

The bound can be expressed as follows:
\begin{align}
b(\bm{x})
=\normOp{\sum^{L}_{l=1}r_{l}\E^{(l)}(x^{(l)})},
\end{align}
where $\normOp{\cdot}$ denotes the operator norm defined as
\begin{align}
\normOp{A}=\sup_{\ket{\psi}:\norm{\ket{\psi}}\neq0}\frac{\norm{A\ket{\psi}}}{\norm{\ket{\psi}}}
\end{align}
for an arbitrary linear operator $A\in\linear(\hil)$.
For a positive operator, the operator norm is equal to the largest eigenvalue.
Thus, $b(\bm{x})$ is given by the largest eigenvalue of $\sum^{L}_{l=1}r_{l}\E^{(l)}(x^{(l)})$.

As an example, for every $l\in[L]_{+}$, suppose that $\Omega_{l}=\{x^{(l)}_{0},\ldots,x^{(l)}_{d-1}\}$, and each element of $\E_{l}$ is given by
\begin{align}
\E^{(l)}(x^{(l)}_{i})&=\ketbra{e^{(l)}_{i}}{e^{(l)}_{i}},
\end{align}
where $\{\ket{e^{(l)}_{i}}\}_{i\in[d-1]}$ is an orthonormal basis.
The symbol $[\cdot]$ denotes the subset of $\mathbb{Z}$ defined as $[N]:=\{i\in\mathbb{Z}:0\leq i\leq N\}$ for an arbitrary integer $N$.
If $L=2$ and $r_{1}=r_{2}=1/2$, the bound is represented as
\begin{align}
b(x^{(1)}_{i},x^{(2)}_{j})
=\frac{1}{2}\left(1+\abs{\braket{e^{(1)}_{i}|e^{(2)}_{j}}}\right) \label{eq:fgur_bound}
\end{align}
for every $i,j\in[d-1]$ \cite{PhysRevA.104.032424,PhysRevA.109.022408}.
In particular, if $\{\ket{e^{(1)}_{i}}\}_{i\in[d-1]}$ and $\{\ket{e^{(2)}_{i}}\}_{i\in[d-1]}$ are mutually unbiased, i.e., the equality $\abs{\braket{e^{(1)}_{i}|e^{(2)}_{j}}}=1/\sqrt{d}$ holds for every $i,j\in[d-1]$, the bound becomes
\begin{align}
b(x^{(1)}_{i},x^{(2)}_{j})
=\frac{1}{2}\left(1+\frac{1}{\sqrt{d}}\right).
\end{align}
Hence, the unavoidable trade-off between $p_{1}(x^{(1)}_{i}|\rho)$ and $p_{2}(x^{(2)}_{j}|\rho)$ exists for all $i,j\in[d-1]$.

\section{Quantum Testers} \label{sec:tester}
We use the framework for measuring quantum channels developed in Refs. \citen{PhysRevLett.101.180501} and \citen{PhysRevA.77.062112} to generalize the FGURs.

Let $\hil_{\sysi}$ and $\hil_{\syso}$ be Hilbert spaces with dimensions $d_{\sysi}$ and $d_{\syso}$, respectively.
Quantum processes that transform states in $\state(\hil_{\sysi})$ to states in $\state(\hil_{\syso})$ are identified with (quantum) channels, which are completely positive and trace preserving linear maps from $\linear(\hil_{\sysi})$ to $\linear(\hil_{\syso})$.
We denote by $\mathcal{C}(\linear(\hil_{\sysi}),\linear(\hil_{\syso}))$ the set of all channels from $\linear(\hil_{\sysi})$ to $\linear(\hil_{\syso})$.
Let $\Lambda$ be a channel in $\mathcal{C}(\linear(\hil_{\sysi}),\linear(\hil_{\syso}))$.
The information about $\Lambda$ is obtained by inputting a known state into $\Lambda$ and performing a known POVM measurement on the resulting output state.
We refer to these procedures as tests of $\Lambda$.
Specifically, a test of $\Lambda$ is carried out as follows:
\begin{enumerate}
\item Prepare a quantum state $\rho$ on $\hil_{\sysa}\otimes\hil_{\sysi}$, where $\hil_{\sysa}$ is an ancillary system with dimension $d_{\sysa}$.
\item Apply $\id_{\sysa}\otimes\Lambda$ to $\rho$, where $\id_{\sysa}$ is the identity map on $\linear(\hil_{\sysa})$.
\item Perform a measurement of the output state, i.e., $(\id_{\sysa}\otimes\Lambda)(\rho)$, using a POVM $\E=\{\E(x)\}_{x\in\Omega}$.
\end{enumerate}
Thus, the test is defined by the pair $(\rho,\E)$.
The probability of obtaining an outcome $x\in\Omega$ is calculated as follows:
\begin{align}
p(x|\Lambda)
&=\tr\left[\E(x)(\id_{\sysa}\otimes\Lambda)(\rho)\right].
\end{align}
We can simplify the right-hand side through straightforward computation.
First, for an arbitrary state $\ket{\Psi}$ in $\hil_{\sysa}\otimes\hil_{\sysi}$, consider the linear operator defined as follows:
\begin{align}
K_{\Psi}=\sum^{d_{\sysa}-1}_{a=0}\sum^{d_{\sysi}-1}_{i=0}\ket{a}\braket{a\otimes i|\Psi}\bra{i},
\end{align}
where $\{\ket{a}\}_{a\in[d_{\sysa}-1]}$ and $\{\ket{i}\}_{i\in[d_{\sysi}-1]}$ are orthonormal bases for $\hil_{\sysa}$ and $\hil_{\sysi}$, respectively.
Let $\Upsilon_{\Psi}$ be the completely positive map defined as
\begin{align}
\Upsilon_{\Psi}(A)=K_{\Psi}AK^{\dagger}_{\Psi}
\end{align}
for all $A\in\linear(\hil_{\sysi})$.
The following equality can be easily verified:
\begin{align}
\ketbra{\Psi}{\Psi}=\left(\Upsilon_{\Psi}\otimes\id_{\sysi}\right)(P'_{+}),
\end{align}
where $P'_{+}$ is the unnormalized maximally entangled state defined as
\begin{align}
P'_{+}=\sum^{d_{\sysi}-1}_{i,j=0}\ketbra{i}{j}\otimes\ketbra{i}{j}.
\end{align}
Suppose that $\rho$ can be written as a convex combination of pure states as follows:
\begin{align}
\rho=\sum^{N}_{n=1}\lambda_{n}\ketbra{\Psi_{n}}{\Psi_{n}}.
\end{align}
Let $\Upsilon_{\rho}$ be the completely positive map defined as
\begin{align}
\Upsilon_{\rho}:=\sum^{N}_{n=1}\lambda_{n}\Upsilon_{\Psi_{n}}.
\end{align}
Clearly, the following equality holds:
\begin{align}
\rho=\left(\Upsilon_{\rho}\otimes\id_{\sysi}\right)(P'_{+}). \label{eq:ups}
\end{align}

Using Eq. (\ref{eq:ups}), the probability of obtaining an outcome $x\in\Omega$ can be rewritten as follows:
\begin{align}
p(x|\Lambda)
&=\tr\left[\E(x)(\id_{\sysa}\otimes\Lambda)(\rho)\right] \\
&=\tr\left[\E(x)(\id_{\sysa}\otimes\Lambda)(\Upsilon_{\rho}\otimes\id_{\sysi})(P'_{+})\right] \\
&=\tr\left[\E(x)(\Upsilon_{\rho}\otimes\id_{\syso})(\id_{\sysi}\otimes\Lambda)(P'_{+})\right] \\
&=\tr\left[(\Upsilon^{*}_{\rho}\otimes\id_{\syso})(\E(x))(\id_{\sysi}\otimes\Lambda)(P'_{+})\right],
\end{align}
where $\Upsilon^{*}_{\rho}$ is the dual map of $\Upsilon_{\rho}$, which is defined such that $\tr\left[A^{\dagger}\Upsilon_{\rho}(B)\right]=\tr\left[\Upsilon^{*}_{\rho}(A)^{\dagger}B\right]$ holds for every $A\in\linear(\hil_{\sysa})$ and $B\in\linear(\hil_{\sysi})$.
Consequently, by letting
\begin{align}
\T(x)&=(\Upsilon^{*}_{\rho}\otimes\id_{\syso})(\E(x)), \label{eq:def_pc_eff} \\
J_{\Lambda}&=(\id_{\sysi}\otimes\Lambda)(P'_{+}),
\end{align}
the probability is expressed as follows:
\begin{align}
p(x|\Lambda)=\tr\left[\T(x)J_{\Lambda}\right]. \label{eq:born_ch}
\end{align}
It is important to note that $\T(x)$ depends on $\rho$ and $\E(x)$, though this dependency is not explicitly shown.
All information about the objects used to test $\Lambda$, namely, $\rho$ and $\E$, is captured in $\T=\{\T(x)\}_{x\in\Omega}$.
Thus, we can consider $\T$ as a mathematical representation of the test.
The properties of $\T$ are similar to those of POVMs.
Indeed, each $\T(x)$ is positive because $\Upsilon^{*}_{\rho}$ is completely positive.
However, the sum of all elements does not necessarily equal the identity because the following equality holds:
\begin{align}
\Upsilon^{*}_{\rho}(I_{\sysa})=\rho^{T}_{\sysi}, \label{eq:ups_ad_id}
\end{align}
where $\rho_{\sysi}=\tr_{\sysa}\rho$ and $^{T}$ denotes the transpose with respect to $\{\ket{i}\}_{i\in[d_{\sysi}-1]}$.
Hence, the sum of $\T(x)$ becomes
\begin{align}
\sum_{x\in\Omega}\T(x)=\rho^{T}_{\sysi}\otimes I_{\syso}. \label{eq:norma_ppovm}
\end{align}
In general, a family of positive operators that satisfies Eq. (\ref{eq:norma_ppovm}) with some state on $\hil_{\sysi}$ is referred to as a (quantum) tester\cite{PhysRevLett.101.180501} or process POVM (PPOVM)\cite{PhysRevA.77.062112}.

The Choi--Jamiołkowski isomorphism\cite{JAMIOLKOWSKI1972275,CHOI1975285} states that $J_{\Lambda}$ is positive and that the map $\Lambda\mapsto J_{\Lambda}$ is bijective.
Consequently, we can identify $J_{\Lambda}$ with $\Lambda$.
The operator $J_{\Lambda}$ is known as the Choi operator.

If $d_{\sysa}=d_{\sysi}=1$, without loss of generality, we can assume the equalities
\begin{align}
\rho&=I_{\sysa}\otimes I_{\sysi}, \\
\E(x)&=I_{\sysa}\otimes \E_{\syso}(x), \\
\Lambda(I_{\sysi})&=\rho_{\syso},
\end{align}
where $\rho_{\syso}$ is a quantum state on $\hil_{\syso}$ and $\{\E_{\syso}(x)\}_{x\in\Omega}$ is a POVM on $\hil_{\syso}$.
The tester becomes
\begin{align}
\T=\{I_{\sysi}\otimes\E_{\syso}(x)\}_{x\in\Omega},
\end{align}
which is a POVM, and the Choi operator becomes
\begin{align}
J_{\Lambda}=I_{\sysi}\otimes\rho_{\syso},
\end{align}
which is a state.
Furthermore, the right-hand side of Eq. (\ref{eq:born_ch}) becomes
\begin{align}
\tr\left[\E_{\syso}(x)\rho_{\syso}\right].
\end{align}
Therefore, the framework presented here encompasses measurements of states as special cases, and Eq. (\ref{eq:born_ch}) can be viewed as a generalized Born rule.

\section{Results} \label{sec:result}
\subsection{Generalization of the FGURs}
We can apply Eq. (\ref{eq:fg}) not only to measurements of states but also to tests of channels.
In this context, POVMs and state are replaced by testers and Choi operator, respectively.

Let $\hil_{\sysi}$ and $\hil_{\syso}$ be Hilbert spaces, and $\Lambda$ be a channel in $\mathcal{C}(\linear(\hil_{\sysi}),\linear(\hil_{\syso}))$.
We consider $L$ tests of $\Lambda$ that are defined by $(\rho^{(l)},\{\E^{(l)}(x^{(l)})\}_{x^{(l)}\in\Omega_{l}})$, where $l\in[L]_{+}$ denotes the test number.
For each $l\in[L]_{+}$, $\rho^{(l)}$ is a state in $\state(\hil_{\sysa}\otimes\hil_{\sysi})$, and $\{\E^{(l)}(x^{(l)})\}_{x^{(l)}\in\Omega_{l}}$ is a POVM on $\hil_{\sysa}\otimes\hil_{\syso}$.
Then the tester is given by
\begin{align}
\T^{(l)}
&=\{(\Upsilon^{*}_{\rho^{(l)}}\otimes\id_{\sysi})(\E^{(l)}(x^{(l)}))\}_{x^{(l)}\in\Omega_{l}} \\
&=\{\T^{(l)}(x^{(l)})\}_{x^{(l)}\in\Omega_{l}}.
\end{align}
It is assumed that $\Omega_{l}\cap\Omega_{l'}=\emptyset$ if $l\neq l'$.
Let $\{r_{l}\}_{l\in[L]_{+}}$ be a probability distribution, and suppose that $\Lambda$ is tested using $\T^{(l)}$ with the probability $r_{l}$.
According to Eq. (\ref{eq:born_ch}), the probability of obtaining an outcome $x^{(l)}$ is given by
\begin{align}
r_{l}p_{l}(x^{(l)}|\Lambda)=r_{l}\tr\left[\T^{(l)}(x^{(l)})J_{\Lambda}\right].
\end{align}
We consider the following set of inequalities:
\begin{align}
\mathcal{V}
=\left\{\sum^{L}_{l=1}r_{l}p_{l}(x^{(l)}|\Lambda)\leq c(\bm{x}):\bm{x}:=(x^{(1)},\ldots,x^{(L)})\in\prod^{L}_{l=1}\Omega_{l}\right\}. \label{eq:fg_ch}
\end{align}
The bound is defined as
\begin{align}
c(\bm{x})
=\max_{\Lambda}\left\{\sum^{L}_{l=1}r_{l}p_{l}(x^{(l)}|\Lambda)\right\}, \label{eq:b}
\end{align}
where the maximization is taken over all channels in $\mathcal{C}(\linear(\hil_{\sysi}),\linear(\hil_{\syso}))$.
Similar to $\mathcal{U}$, each inequality in $\mathcal{V}$ represents the limit of the probability of obtaining one of the outcomes in a sequence $\bm{x}\in\prod^{L}_{l=1}\Omega_{l}$.
For all $\bm{x}\in\prod^{L}_{l=1}\Omega_{l}$, we also consider the quantity
\begin{align}
t(\bm{x}):=\sum^{L}_{l=1}r_{l}\max_{\Lambda}p_{l}(x^{(l)}|\Lambda),
\end{align}
where the maximization has the same meaning as in Eq. (\ref{eq:b}).
For a combination of outcomes $\bm{x}\in\prod^{L}_{l=1}\Omega_{l}$, if the strict inequality
\begin{align}
c(\bm{x})<t(\bm{x}) \label{eq:triv}
\end{align}
holds, we can conclude, as in Sect. \ref{sec:fg} that an unavoidable trade-off exists between the probabilities
\begin{align}
p_{l}(x^{(1)}|\Lambda),\ldots,p_{L}(x^{(L)}|\Lambda).
\end{align}
In particular, when $t(\bm{x})=1$, a trade-off occurs if $c(\bm{x})<1$.
Therefore, $t(\bm{x})$ can be regarded as a trivial bound.

If $d_{\sysa}=d_{\sysi}=1$, $\mathcal{V}$ and $c(\bm{x})$ become
\begin{align}
\mathcal{V}
=\left\{\sum^{L}_{l=1}r_{l}\tr[\E^{(l)}_{\syso}(x_{l})\rho_{\syso}]\leq c(\bm{x}):\bm{x}\in\prod^{L}_{l=1}\Omega_{l}\right\} \label{eq:fgur_spe}
\end{align}
and
\begin{align}
c(\bm{x})
=\max_{\rho_{\syso}}\left\{\sum^{L}_{l=1}r_{l}\tr[\E^{(l)}_{\syso}(x^{(l)})\rho_{\syso}]\right\},
\end{align}
respectively.
Here, $\{\E^{(l)}_{\syso}(x^{(l)})\}_{x^{(l)}\in\Omega_{l}}$ is a POVM on $\hil_{\syso}$ for each $l\in[L]_{+}$, and $\rho_{\syso}$ is a state on $\hil_{\syso}$.
The maximization is taken over all states on $\hil_{\syso}$.
Equation (\ref{eq:fgur_spe}) corresponds to Eq. (\ref{eq:fg}), making $\mathcal{V}$ a conceptual generalization of $\mathcal{U}$.
The set $\mathcal{V}$ is defined using the testers, which is why we refer to sets of inequalities in the form of Eq. (\ref{eq:fg_ch}) as FGURs for quantum testers.

\subsection{Estimation of the uncertainty bound}
The uncertainty bound $c(\bm{x})$ is defined as an optimization problem, making its exact analytical computation challenging.
However, we can instead calculate an upper bound on the probability associated with a combination of outcomes using the following proposition.
\begin{pro} \label{pro:b_est}
For an arbitrary $\bm{x}\in\prod^{L}_{l=1}\Omega_{l}$, the following inequality holds:
\begin{align}
c(\bm{x})
\leq \ctil(\bm{x}), \label{eq:b_ineq}
\end{align}
where the right-hand side is defined as
\begin{align}
\ctil(\bm{x})
:=d_{\sysi}\normOp{\sum^{L}_{l=1}r_{l}\T^{(l)}(x^{(l)})}.
\end{align}
Equation (\ref{eq:b_ineq}) becomes equality if $\tr_{\syso}\ketbra{\Psi_{0}}{\Psi_{0}}=I_{\sysi}/d_{\sysi}$, where $\ket{\Psi_{0}}$ denotes an eigenstate of $\sum^{L}_{l=1}r_{l}\T^{(l)}(x^{(l)})$ corresponding to the largest eigenvalue.
\end{pro}
\begin{proof}
The objective function in Eq. (\ref{eq:b}) can be bounded from above as follows:
\begin{align}
\sum^{L}_{l=1}r_{l}\sum_{x^{(l)}\in\mathcal{X}_{l}}p_{l}(x^{(l)}|\Lambda)
&=d_{\sysi}\tr\left[\sum^{L}_{l=1}r_{l}\T^{(l)}(x^{(l)})\frac{J_{\Lambda}}{d_{\sysi}}\right] \\
&=d_{\sysi}\sum^{d_{\sysi}d_{\syso}-1}_{i=0}\mu_{i}\braket{\mu_{i}|\sum^{L}_{l=1}r_{l}\T^{(l)}(x^{(l)})|\mu_{i}} \\
&\leq d_{\sysi}\sum^{d_{\sysi}d_{\syso}-1}_{i=0}\mu_{i}\max_{i}\braket{\mu_{i}|\sum^{L}_{l=1}r_{l}\T^{(l)}(x^{(l)})|\mu_{i}} \\
&=d_{\sysi}\max_{i}\braket{\mu_{i}|\sum^{L}_{l=1}r_{l}\T^{(l)}(x^{(l)})|\mu_{i}} \\
&\leq d_{\sysi}\max_{\ket{\Psi}}\braket{\Psi|\sum^{L}_{l=1}r_{l}\T^{(l)}(x^{(l)})|\Psi} \\
&=d_{\sysi}\normOp{\sum^{L}_{l=1}r_{l}\T^{(l)}(x^{(l)})},
\end{align}
where the maximization on the right-hand side of the second inequality is performed over all unit vectors in $\hil_{\sysi}\otimes\hil_{\syso}$.
To derive the second equality, we use the fact that $J_{\Lambda}/d$ is a state on $\hil_{\sysi}\otimes\hil_{\syso}$, as stated by the Choi--Jamiołkowski isomorphism, along with the spectral decomposition $J_{\Lambda}/d_{\sysi}=\sum^{d_{\sysi}d_{\syso}-1}_{i=0}\mu_{i}\ketbra{\mu_{i}}{\mu_{i}}$, where $\mu_{i}$ are the eigenvalues and $\ket{\mu_{i}}$ are the corresponding eigenstates.
The final equality follows because the inequality
\begin{align}
\braket{\Psi|\sum^{L}_{l=1}r_{l}\T^{(l)}(x^{(l)})|\Psi}
\leq\normOp{\sum^{L}_{l=1}r_{l}\T^{(l)}(x^{(l)})}
\end{align}
holds for any $\ket{\Psi}\in\hil_{\sysi}\otimes\hil_{\syso}$, and this inequality becomes an equality if $\ket{\Psi}$ is an eigenvector corresponding to the largest eigenvalue of $\sum^{L}_{l=1}r_{l}\T^{(l)}(x^{(l)})$.
Thus, Eq. (\ref{eq:b_ineq}) follows.

The Choi--Jamiołkowski isomorphism guarantees the existence of a channel $\Lambda_{0}\in\mathcal{C}(\linear(\hil_{\sysi}),\linear(\hil_{\syso}))$ such that $J_{\Lambda_{0}}/d_{\sysi}=\ketbra{\Psi_{0}}{\Psi_{0}}$ if $\tr_{\syso}\ketbra{\Psi_{0}}{\Psi_{0}}=I_{\sysi}/d_{\sysi}$.
In this case, the following equality holds:
\begin{align}
\sum^{L}_{l=1}r_{l}\sum_{x^{(l)}\in\mathcal{X}_{l}}p_{l}(x^{(l)}|\Lambda_{0})
=d_{\sysi}\normOp{\sum^{L}_{l=1}r_{l}\T^{(l)}(x^{(l)})}.
\end{align}
Thus, the last assertion of the proposition follows.
\end{proof}
The upper bound $\ctil(\bm{x})$ can be used to quantify the trade-off instead of $c(\bm{x})$ if $\ctil(\bm{x})<t(\bm{x})$ because it implies $c(\bm{x})\leq \ctil(\bm{x})<t(\bm{x})$.
Hence, we consider the following set of inequalities:
\begin{align}
\vtil
=\left\{\sum^{L}_{l=1}r_{l}p_{l}(x^{(l)}|\Lambda)\leq \ctil(\bm{x}):\bm{x}\in\prod^{L}_{l=1}\Omega_{l}\right\}. \label{eq:fgur_ch_est}
\end{align}
If $d_{\sysa}=d_{\sysi}=1$, i.e., the tests correspond to measurements of states, then the equality $\vtil=\mathcal{V}$ holds because Eq. (\ref{eq:b_ineq}) becomes an equality.
Therefore, in this case, $\vtil$ can be identified with $\mathcal{U}$, similar to $\mathcal{V}$.
Thus, although $\vtil$ is formulated using the upper bound $\ctil(\bm{x})$, it is sufficiently strong to serve as a generalization of the conventional FGURs.
We also refer to sets of inequalities in the form given by Eq. (\ref{eq:fgur_ch_est}) as FGURs for quantum testers.

To illustrate when Proposition \ref{pro:b_est} is effectively applied, we consider the following example.
\begin{exmp}
Suppose that $d_{\sysa}=1,L=2,r_{1}=r_{2}=1/2$ and for each $l\in[2]_{+}$, the input state is given by
\begin{align}
\rho^{(l)}=I_{\sysa}\otimes\ketbra{\psi^{(l)}}{\psi^{(l)}},
\end{align}
where $\ketbra{\psi^{(l)}}{\psi^{(l)}}$ is a state on $\hil_{\sysi}$.
We also assume that $\Omega_{l}=\{x^{(l)}_{0},\ldots,x^{(l)}_{d_{\syso}-1}\}$, and each element of $\E^{(l)}$ is given as follows:
\begin{align}
\E^{(l)}(x^{(l)}_{i})=I_{\sysa}\otimes\ketbra{e^{(l)}_{i}}{e^{(l)}_{i}},
\end{align}
where $\{\ket{e^{(l)}_{i}}\}_{i\in[d_{\syso}-1]}$ is an orthonormal basis for all $l\in[2]_{+}$.
In this case, the equality $\Upsilon^{*}_{\rho^{(l)}}(I_{\sysa})=\ketbra{\psi^{(l)}}{\psi^{(l)}}^{T}$ holds.
Thus, each element of the testers becomes
\begin{align}
\T^{(l)}(x^{(l)}_{i})=\ketbra{\psi^{(l)}}{\psi^{(l)}}^{T}\otimes\ketbra{e^{(l)}_{i}}{e^{(l)}_{i}}.
\end{align}
The upper bound is calculated as
\begin{align}
\ctil(x^{(1)}_{i},x^{(2)}_{j})
=\frac{d_{\sysi}}{2}\left(1+\abs{\braket{\psi^{(1)}|\psi^{(2)}}}\abs{\braket{e^{(1)}_{i}|e^{(2)}_{j}}}\right)
\end{align}
for every $i,j\in[d_{\syso}-1]$.
Clearly, if $d_{\sysi}\geq2$, the right-hand side is greater than or equal to $1$.
This implies that $\ctil(x^{(1)}_{i},x^{(2)}_{j})$ becomes meaningless because $\ctil(x^{(1)}_{i},x^{(2)}_{j})\geq t(x^{(1)}_{i},x^{(2)}_{j})$ holds in such cases.
\end{exmp}

Based on this example, it is useful to derive a condition which ensures that $\ctil(\bm{x})\leq1$ holds.
\begin{pro} \label{pro:est_le_1}
If $\tr_{\sysa}\rho^{(l)}=I_{\sysi}/d_{\sysi}$ for all $l\in[L]_{+}$, the inequality $\ctil(\bm{x})\leq1$ holds for an arbitrary $\bm{x}\in\prod^{L}_{l=1}\Omega_{l}$.
\end{pro}
To show this proposition, we provide the following lemma.
\begin{lem} \label{lem:tester_povm}
Let $\rho$ be a state on $\hil_{\sysa}\otimes\hil_{\sysi}$, and $\{\E(x)\}_{x\in\Omega}$ be a POVM on $\hil_{\sysa}\otimes\hil_{\syso}$.
The family of positive operators $d_{\sysi}\T:=\{d_{\sysi}\left(\Upsilon^{*}_{\rho}\otimes\id_{\syso}\right)(\E(x))\}_{x\in\Omega}$ is a POVM on $\hil_{\sysi}\otimes\hil_{\syso}$ if and only if $\rho_{\sysi}=I_{\sysi}/d_{\sysi}$.
\end{lem}
\begin{proof}
The following equality holds:
\begin{align}
\sum_{x\in\Omega}d_{\sysi}\left(\Upsilon^{*}_{\rho}\otimes\id_{\syso}\right)(\E(x))
=d_{\sysi}\rho^{T}_{\sysi}\otimes I_{\syso}.
\end{align}
Hence, the family $d_{\sysi}\T$ is a POVM if and only if $d_{\sysi}\rho^{T}_{\sysi}=I_{\sysi}$.
\end{proof}

Now, let us proceed with the proof of Proposition \ref{pro:est_le_1}.
\begin{proof}[Proof of Proposition \ref{pro:est_le_1}]
The following chain of inequalities holds:
\begin{align}
\ctil(\bm{x})
&\leq\sum^{L}_{l=1}r_{l}\normOp{d_{\sysi}\T^{(l)}(x^{(l)})} \\
&\leq\sum^{L}_{l=1}r_{l} \\
&=1,
\end{align}
where the first inequality is obtained using the triangle inequality.
The second inequality follows from Lemma \ref{lem:tester_povm}, which guarantees that all $d_{\sysi}\T^{(l)}$ are POVMs.
\end{proof}

\subsection{FGURs for tests with a maximally entangled input state}
A straightforward scenario that satisfies the assumption of Proposition \ref{pro:est_le_1} is to choose the maximally entangled state $P_{+}:=P'_{+}/d_{\sysi}$ as $\rho^{(l)}$.
This approach is commonly used in ancilla-assisted quantum process tomography\cite{PhysRevLett.91.047902,PhysRevLett.90.193601}.

Therefore, from now on, we assume $\hil_{\sysa}=\hil_{\sysi}=\hil$, where $\hil$ is a Hilbert space of the dimension $d$, and $\rho^{(l)}=P_{+}$ for all $l\in[L]_{+}$.
Owing to Proposition \ref{pro:est_le_1}, $\ctil(\bm{x})\leq1$ is ensured.
As mentioned earlier, to assert the existence of an unavoidable trade-off for $\bm{x}\in\prod^{L}_{l=1}\Omega_{l}$, the strict inequality $\ctil(\bm{x})<t(\bm{x})$ is required.
\begin{exmp}
Suppose that $L=2,r_{1}=r_{2}=1/2$, and the outcome set is given by $\Omega_{l}=\{x^{(l)}_{i,j}\}$, where $i\in[d-1]$ and $j\in[d_{\syso}-1]$.
Additionally, suppose that each element of $\E^{(l)}$ is expressed as follows:
\begin{align}
\E^{(l)}(x^{(l)}_{i,j})=\ketbra{e^{(l)}_{i}}{e^{(l)}_{i}}\otimes\ketbra{f^{(l)}_{j}}{f^{(l)}_{j}},
\end{align}
where $\{\ket{e^{(l)}_{i}}\}_{i\in[d-1]}$ and $\{\ket{f^{(l)}_{j}}\}_{j\in[d_{\syso}-1]}$ are orthonormal bases for all $l\in[2]_{+}$.
The equality $\Upsilon^{*}_{\rho^{(l)}}=1/d$ holds.
Thus, each element of the testers becomes
\begin{align}
\T^{(l)}(x^{(l)}_{i,j})=\frac{1}{d}\ketbra{e^{(l)}_{i}}{e^{(l)}_{i}}\otimes\ketbra{f^{(l)}_{j}}{f^{(l)}_{j}}.
\end{align}
The upper bound is calculated as
\begin{align}
\ctil(x^{(1)}_{i(1),j(1)},x^{(2)}_{i(2),j(2)})
=\frac{1}{2}\left(1+\abs{\braket{e^{(1)}_{i(1)}|e^{(2)}_{i(2)}}}\abs{\braket{f^{(1)}_{j(1)}|f^{(2)}_{j(2)}}}\right)
\end{align}
for every $(i(1),j(1)),(i(2),j(2))\in[d-1]\times[d_{\syso}-1]$.
The trivial bound is calculated as follows:
\begin{align}
t(x^{(1)}_{i(1),j(1)},x^{(2)}_{i(2),j(2)})
&=\frac{1}{2d}\sum^{2}_{l=1}\max_{\Lambda}\tr\left[\left(\ketbra{e^{(l)}_{i(l)}}{e^{(l)}_{i(l)}}\otimes\ketbra{f^{(l)}_{j(l)}}{f^{(l)}_{j(l)}}\right)J_{\Lambda}\right] \\
&=\frac{1}{2d}\sum^{2}_{l=1}\max_{\Lambda}\braket{f^{(l)}_{j(l)}|\Lambda\left(\ketbra{(e^{(l)}_{i(l)})^{*}}{(e^{(l)}_{i(l)})^{*}}\right)|f^{(l)}_{j(l)}} \\
&=\frac{1}{d},
\end{align}
where $\ket{(\cdot)^{*}}$ denotes the complex conjugate with respect to $\{\ket{i}\}_{i\in[d-1]}$.
Each maximization in the right-hand side of the second equality is achieved by a channel $\Lambda^{(l)}_{i(l),j(l)}$ that acts as follows:
\begin{align}
\Lambda^{(l)}_{i(l),j(l)}\left(\ketbra{(e^{(l)}_{i(l)})^{*}}{(e^{(l)}_{i(l)})^{*}}\right)=\ketbra{f^{(l)}_{j(l)}}{f^{(l)}_{j(l)}}.
\end{align}
For example, this equality is realized by the depolarizing channel defined by $\Lambda^{(l)}_{i(l),j(l)}(\rho)=\ketbra{f^{(l)}_{j(l)}}{f^{(l)}_{j(l)}}$ for all $\rho\in\state(\hil_{\sysi})$, or, if $\hil_{\syso}=\hil$, a unitary channel whose unitary operator maps $\ket{(e^{(l)}_{i(l)})^{*}}$ to $\ket{f^{(l)}_{j(l)}}$.
If $d\geq2$, $\ctil(x^{(1)}_{i(1),j(1)},x^{(2)}_{i(2),j(2)})$ provides no meaningful information because $t(x^{(1)}_{i(1),j(1)},x^{(2)}_{i(2),j(2)})\leq \ctil(x^{(1)}_{i(1),j(1)},x^{(2)}_{i(2),j(2)})$ holds.
\end{exmp}
This example motivates us to consider settings such that $t(\bm{x})=1$, namely, cases where the probability distributions can be deterministic if the channel is chosen for each test.
A straightforward approach is to use orthonormal bases consisting of maximally entangled states as the POVMs $\E^{(l)}$.
We refer to these bases as maximally entangled bases (MEBs).
An example of an MEB is the Bell basis, which is used in many protocols\cite{nielsen2010quantum}.
Let $\hil_{\syso}=\hil$, and let $\{\ket{\Psi^{(l)}_{i}}\}_{i\in[d^{2}-1]}$ be an MEB of $\hil_{\sysa}\otimes\hil_{\syso}=\hil\otimes\hil$ for all $l\in[L]_{+}$.
Suppose that for every $l\in[L]_{+}$, the outcome set is given by $\Omega_{l}=\{x^{(l)}_{0},\ldots,x^{(l)}_{d^{2}-1}\}$, and each element of $\E^{(l)}$ is expressed as
\begin{align}
\E^{(l)}(x^{(l)}_{i})=\ketbra{\Psi^{(l)}_{i}}{\Psi^{(l)}_{i}}. \label{eq:povm_meb}
\end{align}
Each element of the testers becomes
\begin{align}
\T^{(l)}(x^{(l)}_{i})=\frac{1}{d}\ketbra{\Psi^{(l)}_{i}}{\Psi^{(l)}_{i}}.
\end{align}
Let $i(l)$ be an integer in $[d^{2}-1]$ for each $l\in[L]_{+}$.
For an arbitrary $(x^{(1)}_{i(1)},\ldots,x^{(L)}_{i(L)})\in\prod^{L}_{l=1}\Omega_{l}$, the trivial bound is calculated as follows:
\begin{align}
t(x^{(1)}_{i(1)},\ldots,x^{(L)}_{i(L)})
&=\sum^{L}_{l=1}r_{l}\max_{\Lambda}\braket{\Psi^{(l)}_{i(l)}|\frac{J_{\Lambda}}{d}|\Psi^{(l)}_{i(l)}} \\
&=1.
\end{align}
The second equality follows because for all $l\in[L]_{+}$, there is a unitary operator such that $\ket{\Psi^{(l)}_{i(l)}}=(I_{\sysa}\otimes U^{(l)}_{i(l)})\ket{\Psi_{+}}$, where $\ket{\Psi_{+}}=(1/\sqrt{d})\sum^{d-1}_{i=0}\ket{i\otimes i}$, and hence every maximum is achieved by the unitary channel defined as $\Lambda^{(l)}_{i(l)}(\cdot)=U^{(l)}_{i(l)}(\cdot)(U^{(l)}_{i(l)})^{\dagger}$.

Thus, using MEBs for the POVMs guarantees $t(x^{(1)}_{i(1)},\ldots,x^{(L)}_{i(L)})=1$.
In addition to the settings assumed in the previous paragraph, suppose $L=2$ and $r_{1}=r_{2}=1/2$.
The following proposition provides a specific formula for the upper bound.
\begin{pro}
For all $i,j\in[d^{2}-1]$, the upper bound is calculated as follows:
\begin{align}
\ctil(x^{(1)}_{i},x^{(2)}_{j})
=\frac{1}{2}\left(1+\abs{\braket{\Psi^{(1)}_{i}|\Psi^{(2)}_{j}}}\right). \label{eq:est_meb}
\end{align}
\end{pro}
\begin{proof}
The upper bound is defined as
\begin{align}
\ctil(x^{(1)}_{i},x^{(2)}_{j})
=\frac{1}{2}\normOp{\ketbra{\Psi^{(1)}_{i}}{\Psi^{(1)}_{i}}+\ketbra{\Psi^{(2)}_{j}}{\Psi^{(2)}_{j}}}.
\end{align}
The largest eigenvalue of $\ketbra{\Psi^{(1)}_{i}}{\Psi^{(1)}_{i}}+\ketbra{\Psi^{(2)}_{j}}{\Psi^{(2)}_{j}}$ is given by $1+\abs{\braket{\Psi^{(1)}_{i}|\Psi^{(2)}_{j}}}$.
Hence, Eq. (\ref{eq:est_meb}) holds.
\end{proof}
Therefore, the FGUR for the testers is given by:
\begin{align}
&\vtil \nonumber \\
&=\left\{\frac{1}{2}p_{1}(x^{(1)}_{i}|\Lambda)+\frac{1}{2}p_{2}(x^{(2)}_{j}|\Lambda)
\leq\frac{1}{2}\left(1+\abs{\braket{\Psi^{(1)}_{i}|\Psi^{(2)}_{j}}}\right):i,j\in[d^{2}-1]\right\}. \label{eq:fgur_me}
\end{align}
Every inequality in $\vtil$ implies that the unavoidable trade-off between $p_{1}(x^{(1)}_{i}|\Lambda)$ and $p_{2}(x^{(2)}_{j}|\Lambda)$ exists if $\abs{\braket{\Psi^{(1)}_{i}|\Psi^{(2)}_{j}}}<1$.
Although $\ctil(x^{(1)}_{i},x^{(2)}_{j})$ is in general an upper bound, the right-hand side of Eq. (\ref{eq:est_meb}) coincides with the maximum, i.e., $c(x^{(1)}_{i},x^{(2)}_{j})$, if $d=2$.
\begin{pro} \label{pro:opt}
If $d=2$, then for all $i,j\in[d^{2}-1]$, the bound is calculated as follows:
\begin{align}
c(x^{(1)}_{i},x^{(2)}_{j})
=\frac{1}{2}\left(1+\abs{\braket{\Psi^{(1)}_{i}|\Psi^{(2)}_{j}}}\right).
\end{align}
The maximum in Eq. (\ref{eq:b}) is achieved by the unitary channel $\Lambda_{x^{(1)}_{i},x^{(2)}_{j}}(\cdot)=U_{x^{(1)}_{i},x^{(2)}_{j}}(\cdots)U_{x^{(1)}_{i},x^{(2)}_{j}}^{\dagger}$.
Here, $U_{x^{(1)}_{i},x^{(2)}_{j}}$ is defined by
\begin{align}
&U_{x^{(1)}_{i},x^{(2)}_{j}} \nonumber \\
&=
\begin{cases}
\frac{1}{\sqrt{2+\abs{\iProdHS{U^{(1)}_{i}}{U^{(2)}_{j}}}}}\left(U^{(1)}_{i}+\frac{\iProdHS{U^{(2)}_{j}}{U^{(1)}_{i}}}{\abs{\iProdHS{U^{(1)}_{i}}{U^{(2)}_{j}}}}U^{(2)}_{j}\right) & (\iProdHS{U^{(1)}_{i}}{U^{(2)}_{j}}\neq0), \\
U^{(1)}_{i} & (\iProdHS{U^{(1)}_{i}}{U^{(2)}_{j}}=0),
\end{cases}
\label{eq:opt_uni}
\end{align}
where $U^{(1)}_{i}$ and $U^{(2)}_{j}$ are the unitary operators given by $\ket{\Psi^{(1)}_{i}}=(I\otimes U^{(1)}_{i})\ket{\Psi_{+}}$ and $\ket{\Psi^{(2)}_{j}}=(I\otimes U^{(2)}_{j})\ket{\Psi_{+}}$, respectively, and $\iProdHS{U^{(1)}_{i}}{U^{(2)}_{j}}$ denotes the Hilbert--Schmidt inner product, i.e., $\iProdHS{U^{(1)}_{i}}{U^{(2)}_{j}}=\tr\left[(U^{(1)}_{i})^{\dagger}U^{(2)}_{j}\right]$.
\end{pro}
\begin{proof}
The bound is given by
\begin{align}
c(x^{(1)}_{i},x^{(2)}_{j})
=\frac{1}{2}\max_{\Lambda}\tr\left[\left(\ketbra{\Psi^{(1)}_{i}}{\Psi^{(1)}_{i}}+\ketbra{\Psi^{(2)}_{j}}{\Psi^{(2)}_{j}}\right)\frac{J_{\Lambda}}{d}\right].
\end{align}
Therefore, it suffices to demonstrate that if $d=2$, $U_{x^{(1)}_{i},x^{(2)}_{j}}$ is a unitary operator, and the following equality holds:
\begin{align}
&\tr\left[\left(\ketbra{\Psi^{(1)}_{i}}{\Psi^{(1)}_{i}}+\ketbra{\Psi^{(2)}_{j}}{\Psi^{(2)}_{j}}\right)(\id\otimes\Lambda_{x^{(1)}_{i},x^{(2)}_{j}})\left(P_{+}\right)\right] \nonumber \\
&=1+\abs{\braket{\Psi^{(1)}_{i}|\Psi^{(2)}_{j}}}. \label{eq:b_opt}
\end{align} 
If $\iProdHS{U^{(1)}_{i}}{U^{(2)}_{j}}=0$, these properties are readily established.
Thus, we can assume $\iProdHS{U^{(1)}_{i}}{U^{(2)}_{j}}\neq0$.

Through direct computation, the following equality can be obtained:
\begin{align}
&U_{x^{(1)}_{i},x^{(2)}_{j}}^{\dagger}U_{x^{(1)}_{i},x^{(2)}_{j}} \nonumber \\
&=\frac{1}{2+\abs{\iProdHS{U^{(1)}_{i}}{U^{(2)}_{j}}}}\left(2I_{\sysi}+2\Real\left[\frac{\iProdHS{U^{(2)}_{j}}{U^{(1)}_{i}}}{\abs{\iProdHS{U^{(1)}_{i}}{U^{(2)}_{j}}}}(U^{(1)}_{i})^{\dagger}U^{(2)}_{j}\right]\right),
\end{align}
where $\Real[\cdot]$ is the real part of its argument, defined as $\Real[\cdot]=((\cdot)+(\cdot)^{\dagger})/2$.
Thus, $U_{x^{(1)}_{i},x^{(2)}_{j}}$ is unitary if and only if
\begin{align}
2\Real\left[\iProdHS{U^{(2)}_{j}}{U^{(1)}_{i}}(U^{(1)}_{i})^{\dagger}U^{(2)}_{j}\right]
=\abs{\iProdHS{U^{(1)}_{i}}{U^{(2)}_{j}}}^{2}I_{\sysi}. \label{eq:uni_iff}
\end{align}
Let $\{u_{n}\}_{n\in[d-1]}$ be the eigenvalues of $(U^{(1)}_{i})^{\dagger}U^{(2)}_{j}$.
Equation (\ref{eq:uni_iff}) is equivalent to the condition that the following equality holds for all $m\in[d-1]$:
\begin{align}
\abs{\sum_{n\neq m}u_{n}}^{2}=1
\end{align}
which always holds if $d=2$ (but does not necessarily hold if $d>2$).
Therefore, $U_{x^{(1)}_{i},x^{(2)}_{j}}$ is unitary if $d=2$.

If $d=2$, a simple calculation shows that the following equality holds:
\begin{align}
&(I_{\sysi}\otimes U_{x^{(1)}_{i},x^{(2)}_{j}})\ket{\Psi_{+}} \nonumber \\
&=\frac{1}{\sqrt{2\left(1+\abs{\braket{\Psi^{(1)}_{i}|\Psi^{(2)}_{j}}}\right)}}\left(\ket{\Psi^{(1)}_{i}}+\frac{\braket{\Psi^{(2)}_{j}|\Psi^{(1)}_{i}}}{\abs{\braket{\Psi^{(1)}_{i}|\Psi^{(2)}_{j}}}}\ket{\Psi^{(2)}_{j}}\right).
\end{align}
The right-hand side is the eigenstate of $\ketbra{\Psi^{(1)}_{i}}{\Psi^{(1)}_{i}}+\ketbra{\Psi^{(2)}_{j}}{\Psi^{(2)}_{j}}$ corresponding to the eigenvalue $1+\abs{\braket{\Psi^{(1)}_{i}|\Psi^{(2)}_{j}}}$.
Hence, Eq. (\ref{eq:b_opt}) holds.
\end{proof}

Therefore, if $d=2$, the FGUR for the testers is given by
\begin{align}
\mathcal{V}
=\left\{\frac{1}{2}p_{1}(x^{(1)}_{i}|\Lambda)+\frac{1}{2}p_{2}(x^{(2)}_{j}|\Lambda)
\leq\frac{1}{2}\left(1+\abs{\braket{\Psi^{(1)}_{i}|\Psi^{(2)}_{j}}}\right):i,j\in[3]\right\}. \label{eq:fgur_two_mebs}
\end{align}

As stated in Sect. \ref{sec:fg}, measurements with mutually unbiased bases lead to an unavoidable trade-off for arbitrary combinations of outcomes.
Therefore, it is interesting to ask whether there exist pairs of MEBs that are also mutually unbiased.
In Ref. \cite{PhysRevA.65.032320}, the orthonormal bases
\begin{align}
\biggl\{\ket{\Phi^{(1)}_{0}}&=\frac{1}{\sqrt{2}}\left(\ket{0\otimes x_{0}}+\ket{1\otimes x_{1}}\right), \nonumber \\
\ket{\Phi^{(1)}_{1}}&=\frac{1}{\sqrt{2}}\left(\ket{0\otimes x_{0}}-\ket{1\otimes x_{1}}\right), \nonumber \\
\ket{\Phi^{(1)}_{2}}&=\frac{1}{\sqrt{2}}\left(\ket{0\otimes x_{1}}+\ket{1\otimes x_{0}}\right), \nonumber \\
\ket{\Phi^{(1)}_{3}}&=\frac{1}{\sqrt{2}}\left(\ket{0\otimes x_{1}}-\ket{1\otimes x_{0}}\right)\biggr\}
\end{align}
and
\begin{align}
\biggl\{\ket{\Phi^{(2)}_{0}}&=\frac{1}{\sqrt{2}}\left(\ket{y_{0}\otimes 0}+i\ket{y_{1}\otimes 1}\right), \nonumber \\
\ket{\Phi^{(2)}_{1}}&=\frac{1}{\sqrt{2}}\left(\ket{y_{0}\otimes 0}-i\ket{y_{1}\otimes 1}\right), \nonumber \\
\ket{\Phi^{(2)}_{2}}&=\frac{1}{\sqrt{2}}\left(\ket{y_{0}\otimes 1}+i\ket{y_{1}\otimes 0}\right), \nonumber \\
\ket{\Phi^{(2)}_{3}}&=\frac{1}{\sqrt{2}}\left(\ket{y_{0}\otimes 1}-i\ket{y_{1}\otimes 0}\right)\biggr\}
\end{align}
are designed for a two-qubit system.
Here, $\{\ket{0},\ket{1}\}$ denotes an orthonormal basis for each qubit, and $\ket{x_{j}}$ and $\ket{y_{j}}$ are defined as $\ket{x_{j}}=(1/\sqrt{2})(\ket{0}+(-1)^{j}\ket{1})$ and $\ket{y_{j}}=(1/\sqrt{2})(\ket{0}+(-1)^{j}i\ket{1})$, respectively.
These two bases are mutually unbiased, i.e., $\abs{\braket{\Phi^{(1)}_{i}|\Phi^{(2)}_{j}}}=1/2$ for all $i,j\in[3]$.
Therefore, if $\ket{\Psi^{(l)}_{i}}=\ket{\Phi^{(l)}_{i}}$ for every $l\in[2]$ and $i\in[3]$, then Eq. (\ref{eq:fgur_two_mebs}) becomes
\begin{align}
\mathcal{V}
=\left\{\frac{1}{2}p_{1}(x^{(1)}_{i}|\Lambda)+\frac{1}{2}p_{2}(x^{(2)}_{j}|\Lambda)
\leq\frac{3}{4}:i,j\in[3]\right\}, \label{eq:ur_mumeb}
\end{align}
which indicates that the unavoidable trade-off between $p_{1}(x^{(1)}_{i}|\Lambda)$ and $p_{2}(x^{(2)}_{j}|\Lambda)$ exists for all $i,j\in[3]$.
Therefore, the pair of tests given by $(P_{+}/2,\{\ketbra{\Phi^{(1)}_{i}}{\Phi^{(1)}_{i}}\}_{i\in[3]})$ and $(P_{+}/2,\{\ketbra{\Phi^{(2)}_{j}}{\Phi^{(2)}_{j}}\}_{j\in[3]})$ is a concrete example that illustrates the nontrivial FGUR for testers.

\section{Conclusions and Discussion} \label{sec:conclusion}
Following Ref. \citen{PhysRevResearch.3.023077}, we have extended the FGURs derived in Ref. \citen{doi:10.1126/science.1192065} to apply them not only to measurements of states but also to tests of channels.
Because the bound of the generalized FGURs is defined as the maximum probability associated with a combination of outcomes, we provided an upper bound on this probability as a more accessible alternative.
In the case where tests correspond to measurements of states, the FGURs for testers align with the conventional FGURs.
Thus, our formulations represent genuine generalizations of the FGURs.
We established a necessary condition for the upper bound to be meaningful, and, based on this condition, examined tests using a maximally entangled input state.
As a result, an explicit FGUR for testers was derived, as shown in Eq. (\ref{eq:fgur_me}), for two tests involving MEBs.
Specifically, we introduced a mutually unbiased pair of MEBs for a two-qubit system to provide a concrete example that yields a nontrivial FGUR for testers.

It is worth noting that the left-hand side of the inequality in Eq. (\ref{eq:fgur_me}) can be calculated as
\begin{align}
\frac{1}{2}p_{1}(x^{(1)}_{i}|\Lambda)+\frac{1}{2}p_{2}(x^{(2)}_{j}|\Lambda)
&=\frac{1}{2}\braket{\Psi^{(1)}_{i}|\frac{J_{\Lambda}}{d}|\Psi^{(1)}_{i}}+\frac{1}{2}\braket{\Psi^{(2)}_{j}|\frac{J_{\Lambda}}{d}|\Psi^{(2)}_{j}}.
\end{align}
The operator $J_{\Lambda}/d$ is a state, and $\{\ketbra{\Psi^{(1)}_{i}}{\Psi^{(1)}_{i}}\}_{i\in[d^{2}-1]}$ and $\{\ketbra{\Psi^{(2)}_{j}}{\Psi^{(2)}_{j}}\}_{j\in[d^{2}-1]}$ can be viewed as POVMs.
Hence, Eq. (\ref{eq:fgur_me}) can also be derived using Eq. (\ref{eq:fgur_bound}), which represents the bound of the conventional FGURs.
This implies that existing uncertainty relations for state measurements are applicable to channel tests when they are performed with specific settings, such as a case where a maximally entangled state is used as input states, and MEBs are used as POVM measurements for output states.
However, this does not diminish the value of our analysis because our FGURs for testers are applicable to general tests involving input states that need not be maximally entangled and POVMs that are not necessarily given by MEBs.

Finally, we propose a possible generalization of $\mathcal{V}$.
Instead of $\mathcal{V}$, we can consider the following set of inequalities:
\begin{align}
\mathcal{V}(\mathcal{C}')
=\left\{\sum^{L}_{l=1}r_{l}\sum_{x_{l}\in\mathcal{X}_{l}}p_{l}(x^{(l)}|\Lambda)\leq c(\mathcal{C}',\mathcal{X}):\mathcal{X}:=(\mathcal{X}_{1},\ldots,\mathcal{X}_{L})\in\prod^{L}_{l=1}2^{\Omega_{l}}\right\},
\end{align}
where $\mathcal{C}'$ is a subset of $\mathcal{C}(\linear(\hil_{\sysi}),\linear(\hil_{\syso}))$, and $2^{\Omega_{l}}$ denotes the power set of $\Omega_{l}$ for each $l\in[L]_{+}$.
The bound is defined by
\begin{align}
c(\mathcal{C}',\mathcal{X})
=\max_{\Lambda\in\mathcal{C}'}\sum^{L}_{l=1}r_{l}\sum_{x_{l}\in\mathcal{X}_{l}}p_{l}(x^{(l)}|\Lambda).
\end{align}
Namely, $c(\mathcal{C}',\mathcal{X})$ is the maximum value of the probability associated with outcome subsets $\mathcal{X}_{1},\ldots,\mathcal{X}_{L}$ in $\mathcal{C}'$.
If $\mathcal{C}'=\mathcal{C}(\linear(\hil_{\sysi}),\linear(\hil_{\syso}))$, $\mathcal{V}(\mathcal{C}')$ contains $\mathcal{V}$, and thus is seen as a generalization of $\mathcal{V}$.
Using $\mathcal{V}(\mathcal{C}')$, we can explore any trade-off related to combinations of outcome subsets.
Additionally, restricting the maximization to $\mathcal{C}'$ may make the optimization problem more tractable.
For instance, the uncertainty relation for ancilla-free tests derived in Ref. \citen{doi:10.7566/JPSJ.92.034005} corresponds to the case where $\mathcal{C}'$ is the set of random unitary channels acting on a qubit.
This uncertainty relation is expressed with a straightforward formula, making it easy to compute for any pair of ancilla-free tests.
Besides random unitary channels, there are various other types of channels.
Specifically, entanglement breaking channels\cite{doi:10.1142/S0129055X03001709} play a crucial role in quantum information processing.
Investigating uncertainty relations for such subsets presents an intriguing area for future research.

\section*{Acknowledgment}
I would like to thank T. Miyadera for fruitful comments and discussions.


\begin{thebibliography}{10}

\bibitem{PhysRev.34.163}
H.~P. Robertson, Phys. Rev. \textbf{34}, 163 (1929).

\bibitem{PhysRevLett.50.631}
D. Deutsch, Phys. Rev. Lett. \textbf{50}, 631 (1983).

\bibitem{PhysRevLett.60.1103}
H. Maassen and J.~B.~M. Uffink, Phys. Rev. Lett. \textbf{60}, 1103 (1988).

\bibitem{RevModPhys.89.015002}
P.~J. Coles, M. Berta, M. Tomamichel, and S. Wehner, Rev. Mod. Phys. \textbf{89}, 015002 (2017).

\bibitem{doi:10.1126/science.1192065}
J. Oppenheim and S. Wehner, Science \textbf{330}, 1072 (2010).

\bibitem{PhysRevA.84.052117}
M.~H. Partovi, Phys. Rev. A \textbf{84}, 052117 (2011).

\bibitem{10.1088/1751-8113/46/27/272002}
Z. Puchała, Ł. Rudnicki, and K. Życzkowski, J. Phys. A \textbf{46}, 272002 (2013).

\bibitem{PhysRevLett.111.230401}
S. Friedland, V. Gheorghiu, and G. Gour, Phys. Rev. Lett. \textbf{111}, 230401 (2013).

\bibitem{PhysRevA.89.052115}
Ł. Rudnicki, Z. Puchała, and K. Życzkowski, Phys. Rev. A \textbf{89}, 052115 (2014).

\bibitem{PhysRevResearch.3.023077}
Y. Xiao, K. Sengupta, S. Yang, and G. Gour, Phys. Rev. Res. \textbf{3}, 023077 (2021).

\bibitem{Rastegin_2010}
A.~E. Rastegin, J. Phys. A \textbf{43}, 155302 (2010).

\bibitem{Rastegin2015}
A.~E. Rastegin, Quantum Inf. Process. \textbf{14}, 783 (2015).

\bibitem{PhysRevA.104.032424}
G. Sharma, S. Sazim, and S. Mal, Phys. Rev. A \textbf{104}, 032424 (2021).

\bibitem{PhysRevA.109.022408}
X. Zhou, L.-L. Sun, and S. Yu, Phys. Rev. A \textbf{109}, 022408 (2024).

\bibitem{PhysRevLett.101.180501}
G. Chiribella, G.~M. D'Ariano, and P. Perinotti, Phys. Rev. Lett. \textbf{101}, 180501 (2008).

\bibitem{PhysRevA.77.062112}
M. Ziman, Phys. Rev. A \textbf{77}, 062112 (2008).

\bibitem{JAMIOLKOWSKI1972275}
A.~Jamiołkowski, Rep. Math. Phys. \textbf{3}, 275 (1972).

\bibitem{CHOI1975285}
M.-D. Choi, Linear Algebra Appl. \textbf{10}, 285 (1975).

\bibitem{PhysRevLett.91.047902}
G.~M. D'Ariano and P.~L. Presti, Phys. Rev. Lett. \textbf{91}, 047902 (2003).

\bibitem{PhysRevLett.90.193601}
J.~B.~Altepeter, D.~Branning, E.~Jeffrey, T.~C.~Wei, P.~G.~Kwiat, R.~T.~Thew, J.~L.~O'Brien, M.~A.~Nielsen, and A.~G.~White, Phys. Rev. Lett. \textbf{90}, 193601 (2003).

\bibitem{nielsen2010quantum}
M.~A. Nielsen and I.~L. Chuang, {\em Quantum computation and quantum information} (Cambridge University Press, Cambridge, U.K., 2010) 10th Anniversary Edition.

\bibitem{PhysRevA.65.032320}
J. Lawrence, \v{C}. Brukner, and A. Zeilinger, Phys. Rev. A \textbf{65}, 032320 (2002).

\bibitem{doi:10.7566/JPSJ.92.034005}
T. Kimoto and T. Miyadera, J. Phys. Soc. Jpn. \textbf{92}, 034005 (2023).

\bibitem{doi:10.1142/S0129055X03001709}
M. Horodecki, P.~W. Shor, and M.~B. Ruskai, Rev. Math. Phys. \textbf{15}, 629 (2003).

\end{thebibliography}
\end{document}